\documentclass[lettersize,journal]{IEEEtran}
\usepackage{amsmath,amsfonts}
\usepackage{algorithmic}
\usepackage{algorithm}
\usepackage{array}
\usepackage[caption=false,font=normalsize,labelfont=sf,textfont=sf]{subfig}
\usepackage{textcomp}
\usepackage{stfloats}
\usepackage{url}
\usepackage{verbatim}
\usepackage{graphicx}
\usepackage{cite}

\usepackage{amsthm}

\newcommand{\nop}[1]{}

\newtheorem{theorem}{Theorem}
\newtheorem{corollary}{Corollary}

\newtheorem{proposition}{Proposition}

\begin{document}

\title{Beyond Virtual Delay: Improving Packet Delay Bound in Network Calculus}

\author{Yuming Jiang \\Norwegian University of Science and Technology, Trondheim, Norway \\yuming.jiang@ntnu.no 
\thanks{This work has been submitted to the IEEE for possible publication. Copyright may be transferred without notice, after which this version may no longer be accessible.}
}

\maketitle

\begin{abstract}
In network calculus, a fundamental result is the classical delay bound given by the horizontal deviation between the arrival and service curves. While widely used, the classical bound is derived from the notion of virtual delay. 
For a FIFO system, in this work, we first show that the maximum packet delay is always upper-bounded by the maximum virtual delay, revealing inherent conservatism when applying the virtual-delay-based bound to packet delay. Motivated by this insight, we revisit packet delay analysis and derive a new packet delay bound that requires no assumptions beyond the arrival and service curves. Specializing the new bound to a system with leaky-bucket arrival curve and rate-latency service curve shows strict improvement over the classical bound, which is further demonstrated through a case study in time-sensitive networking (TSN). 
\end{abstract}


\section{Introduction}

Network calculus provides a powerful framework for analyzing worst-case performance guarantees in communication networks \cite{Chang00} \cite{NetCal}\cite{DNC}. By modeling traffic arrivals through arrival curves and system capabilities through service curves, network calculus enables the derivation of bounds on key performance metrics such as backlog and delay. A fundamental result is the classical delay bound derived from the notion of virtual delay and given by the horizontal deviation between the arrival and service curves \cite{Chang00} \cite{NetCal}\cite{DNC}. 
If the only available information is the description of the arrival and service curves, the bound is tight for virtual delay \cite{NetCal}. 
If the system is FIFO, the delay of a packet can be connected to the virtual delay at the time of its arrival. This not only justifies using virtual delay as a proxy for packet delay but also the classical bound 
as a bound on packet delay, e.g. in \cite{Zhao22} and references therein. 

In this work, however, we prove that for a FIFO system, the maximum packet delay is inherently upper-bounded by the maximum virtual delay. This implies that the classical bound can be conservative when applied to packet delay. In addition, it suggests that tighter bounds may be obtained by directly working on packet delay. 
This idea has indeed already been exploited in the literature. 
Specifically, the work \cite{Ehsan19} shows that for a service curve network element, if its transmission rate is known, an improved delay bound can be derived. In \cite{Ehsan23}, it is shown that if the input traffic is also known to have a packet-level arrival curve, the bound can be further improved.  

In this work, we also conduct analysis directly on packet delay without the intermediate use of virtual delay. 
However, unlike \cite{Ehsan19} \cite{Ehsan23}, we investigate whether an improved packet delay bound can be obtained without introducing additional assumptions. Specifically, we ask whether a tighter packet delay bound can be derived using only the description of arrival and service curves so as to preserve generality and abstraction of the bound, like the classical virtual-delay-based bound. We answer this question in the affirmative. The new packet delay bound is exemplified to be strictly tighter than the classical bound in practical systems, which is also illustrated through a case study in time-sensitive networking (TSN). 

In the next section, the system model and the necessary background on network calculus are introduced. Then in Section \ref{sec-3} we prove the relationship between the maximum virtual delay and the maximum packet delay. The new bound is derived and exemplified for a system with leaky-bucket arrival curve and rate-latency service curve in Section \ref{sec-4}. In Section \ref{sec-5}, a TSN case study is conducted to quantitatively compare the bounds. Finally, the conclusion is made in Section \ref{sec-6}.

\section{System Model and the Classical Delay Bound}
\subsection{The System} 
We consider a FIFO system in a packet-switched network. The system may be a single network element or a network of elements. Its input process is a sequence of packets. The input process may consist of multiple flows. 
Packets may be queued within the system but are not dropped. By convention, a packet is considered to have arrived to (respectively departed from) the system if and only if its last bit has arrived (respectively departed). 
For each packet $n = 1, 2, 3, \dots$, let $a(n)$, $d(n)$ and $l(n)$  denote its arrival time, departure time and length (in bits), respectively. 
In addition, let $A(t)$ represent the cumulative input (in bits) entering the system and $A^{*}(t)$ the cumulative output (in bits) from the system, up to time $t$ (excluded) \cite{NetCal}. 

The input has arrival curve $\alpha$ and the system provides to the input service curve $\beta$. By arrival curve and service curve definitions \cite{NetCal}, the input satisfies 
$$
A(t) - A(s) \le \alpha(t-s), \quad \forall 0 \le  s \le t
$$
and the service ensures 
$$
A^{*}(t) \ge A\otimes \beta(t) \equiv \inf_{0 \le s \le t}\{A(s) + \beta(t-s)\}, \quad \forall t \ge 0   
$$
where $\alpha$ and $\beta$ are non-negative wide-sense increasing functions with $\beta(0)=0$, and $\otimes$ is the min-plus convolution operation.

\subsection{Packet Delay, Virtual Delay, and the Classical Bound}

The delay $D(n)$ of packet $n (\ge 1)$ in the system is  
\begin{equation}
D(n) = d(n) -a(n).  
\end{equation}
The virtual delay $\tilde{D}(t)$ at time $t (\ge 0)$ is defined as \cite{NetCal} 
\begin{equation}
\tilde{D}(t) = \inf\{\tau \ge0: A(t) \le A^{*}(t+\tau)\}.
\end{equation}

A fundamental result in network calculus is that the virtual delay $\tilde{D}(t)$ at any time $t$ is upper-bounded by the horizontal deviation between the arrival curve $\alpha$ and the service curve $\beta$ \cite{NetCal}, which is denoted as $h(\alpha, \beta)$ and referred to as {\em the classical delay bound} in this work, i.e., $\forall t \ge 0$, 
\begin{equation}
\tilde{D}(t) \le h(\alpha, \beta) \equiv \sup_{t \ge 0}\inf\{\tau: \alpha(t) \le \beta(t+\tau)\}.
\end{equation}

If $\alpha$ and $\beta$ are the only available information and $\alpha$ is sub-additive with $\alpha(0)=0$, {\em the bound on virtual delay is tight} (Theorem 1.4.4, \cite{NetCal}), i.e. a system with arrival curve $\alpha$ and service curve $\beta$ can be constructed where the bound can be reached by virtual delay.

\section{Relationship between Maximum Packet Delay and Maximum Virtual Delay}\label{sec-3}

Consider the maximum packet delay $D^{max}$ and the maximum virtual delay $\tilde{D}^{max}$, which are, respectively, 
\begin{eqnarray}
    D^{max} &=& \max_{n \ge 1}D(n)  \label{eq-dmax}\\ 
    \tilde{D}^{max} &=& \sup_{t\ge 0} \tilde{D}(t) \label{eq-vdmax} 
\end{eqnarray}

Theorem \ref{th-0} reveals a relationship between the two quantities. In its proof as well as in the rest of the paper, the lower and upper pseudo-inverse functions,  $f^{\downarrow}$ and $f^{\uparrow}$, of a wide-sense increasing function $f$, 
are respectively defined as \cite{Chang00}\cite{JL17}:  
\begin{eqnarray}
f^{\downarrow}(y) &\equiv& \inf \{x\ge 0: f(x) \ge y\} \nonumber\\
f^{\uparrow}(y) &\equiv& \inf \{x\ge 0: f(x) > y\} \nonumber
\end{eqnarray} 
In \cite{JL17}, a collection of their properties can be found, e.g., 
\begin{enumerate}
    \item[(P1)] both $f^{\downarrow}$ and $f^{\uparrow}$ are wide-sense increasing; 
    \item[(P2)] if $f(x) \ge y$, $f^{\downarrow}(y) \le x$;
    \item[(P3)] if $f(x) \le y$, $f^{\uparrow}(y) \ge x$
\end{enumerate}

\begin{theorem}\label{th-0}
    For a FIFO system, if the virtual delay is bounded, i.e. $\tilde{D}(t)< \infty, \forall t \ge 0$,
    the maximum packet delay is upper-bounded by the maximum virtual delay, i.e.,   
$$
D^{max} \le \tilde{D}^{max} . 
$$
\end{theorem}

\begin{proof}
Consider any packet $n (\ge 1)$. Define $a(n)_{+}$ as the right limit of its arrival time $a(n)$, i.e. $a(n)_{+} = a(n)+\epsilon, \epsilon \to 0$.  Clearly, $A(a(n)_{+}) \ge \sum_{k=1}^{n}l(k)$, because at $a(n)$, there may be other packets arriving at the same time. In addition the definition of $\tilde{D}^{max}$ tells $A(a(n)_{+}) \le A^{*}(a(n)_{+}+\tilde{D}^{max})$, so
$$
A^{*}(a(n)_{+}+\tilde{D}^{max}) \ge \sum_{k=1}^{n}l(k)
$$
From property (P2) of the lower pseudo-inverse, we obtain 
$$
(A^{*})^{\downarrow}(\sum_{k=1}^{n}l(k)) \le a(n)_{+}+\tilde{D}^{max}
$$

Note that, we must also have $(A^{*})^{\downarrow}(\sum_{k=1}^{n}l(k)) > d(n)$ because $A^{*}(t)$ reaches $\sum_{k=1}^{n}l(k)$ only after packet $n$ has finished its service due to FIFO. Consequently, we have 
$
a(n)_{+}+\tilde{D}^{max} > d(n) 
$
or by letting $\epsilon \to 0$, $\tilde{D}^{max}\ge d(n) - a(n)$. Since it  holds for all $n$, $\tilde{D}^{max}\ge \max_{n\ge 1}\{d(n) - a(n)\}=D^{max}$. 
\end{proof}

With Theorem \ref{th-0}, Corollary \ref{th0-c1} is immediate. 

\begin{corollary}\label{th0-c1}
    An upper bound on the maximum virtual delay is also an upper bound on the maximum packet delay.
\end{corollary}

Since the horizontal deviation is time-independent, we have $\tilde{D}^{max} \le h(\alpha, \beta)$,  
which, together with Corollary \ref{th0-c1} and $D^{max} \ge D(n), \forall n \ge 1$ from (\ref{eq-dmax}), leads to Corollary \ref{th0-c2}.

\begin{corollary}\label{th0-c2}
The delay $D(n)$ of any packet $n (\ge 1)$ and the maximum packet delay $D^{max}$ are both upper-bounded by the classical delay bound, i.e. the horizontal deviation $h(\alpha, \beta)$.  
\end{corollary}

\section{Improving Packet Delay Bound}\label{sec-4}

Theorem \ref{th-0} implies that the classical virtual-delay bound may overestimate when applied to packet delay. This gap motivates direct analysis at the packet-level to reduce conservatism. 

\subsection{A New Packet Delay Bound}

Theorem \ref{th-1} is from direct analysis at the packet-level, where a new bound is introduced without additional assumptions other than the arrival and service curves. Its potential improvement over the classical bound is demonstrated in Sec. \ref{sec-4-2}. 

\begin{theorem}\label{th-1}
For a FIFO system, if the input has arrival curve $\alpha$ and the system provides to the input service curve $\beta$, then, for any packet $n \ge 1$, its delay $D(n)$ satisfies   
\begin{equation}\label{map-db3}
D(n) \le \sup_{v \ge 0}\{\beta^{\uparrow}(v) - \alpha^{\downarrow}(v+l(n))\}. 
\end{equation}
\end{theorem}

\begin{proof}\footnote{The proof can also be obtained by extending the results in \cite{Chang00}. In particular, step (\ref{eq-6}) extends Lemma 6.2.8 (ii) in \cite{Chang00}, by following the same proof but letting $t=\tau(n)_{+}$ there to obtain (\ref{eq-6}) with $l(n)$ included.  
Step (\ref{eq-7}) follows from Lemma 6.3.2 and Step (\ref{eq-9}) from Theorem 6.3.4 (i) in \cite{Chang00}.} 
Consider packet $n (\ge 1)$. To ease expression, define $L(n) \equiv \sum_{k=0}^{n-1} l(k)$ and  $L(m,n) \equiv \sum_{k=m}^{n-1} l(k)$, $\forall (0 \le) m \le n-1$, where packet 0 is virtual  with $l(0)=a(0)=d(0)=0$.

 First, focus on $a(n)$. For any $(0 \le) m (\le n-1)$, $A(a(m), a(n)_{+}) \ge \sum_{k=m}^{n}l(k)$ because other packets may also arrive at $a(n)$. In addition, from the definition of arrival curve, $A(a(m), a(n)_{+}) \le \alpha(a(n)_{+} -a(m))$. 
So,  
$$
\alpha(a(n)_{+} -a(m)) \ge \sum_{k=m}^{n}l(k) = L(m,n) + l(n).  
$$  
Then, from property (P2) of the lower pseudo-inverse, we get
$$
\alpha^{\downarrow}(L(m,n) + l(n) ) \le a(n)_{+} -a(m) 
$$
 hence for any $(0\le) m (\le n-1)$
\begin{equation}\label{eq-6}
a(n)_{+}  \ge a(m) + \alpha^{\downarrow}(L(m,n) + l(n)). 
\end{equation}   

Next, focus on $d(n)$. The definition of  service curve indicates that there exists some time $s$, $(0 \le s \le d(n))$, such that $A^{*}(d(n)) \ge A(s) + \beta(d(n) -s)$ \cite{NetCal}. In addition, by definition, $A^{*}(t)$ represents the amount of departure up to $t$ (excluded), so $A^{*}(d(n)) \le \sum_{k=0}^{n-1} l(k) = L(n)$. 
Hence,   
\begin{eqnarray}
L(n) &\ge& A(s) + \beta(d(n) -s). \nonumber
\end{eqnarray}
Let $m=\min\{k: a(k-1) < s\}$, which implies $a(m-1) < s \le a(m)$.
With $a(m-1) < s$, at least packet $m-1$ has arrived by $s$, so we must have $A(s) \ge  \sum_{k=0}^{m-1} l(k) = L(m)$. In addition, with $s \le a(m)$, $\beta(d(n) -s) \ge \beta(d(n) -a(m))$ because $\beta$ is non-decreasing. Together, $L(n) \ge L(m) + \beta(d(n) - a(m))$, or
$
\beta(d(n)-a(m)) \le L(m,n). 
$
With property (P3) of the upper pseudo-inverse, we obtain
$$
\beta^{\uparrow}(L(m,n)) \ge d(n) -a(m) 
$$
and hence, 
\begin{eqnarray}
d(n) &\le& a(m) + \beta^{\uparrow}(L(m,n)). \label{eq-7} 
\end{eqnarray} 

Finally, focus on $D(n)= d(n) - a(n)$. Since $a(n)_{+}=a(n) + \epsilon$, letting $\epsilon \to 0$ and combining (\ref{eq-6}) and (\ref{eq-7}), we get 
\begin{eqnarray}
D(n) &\le& \beta^{\uparrow}(L(m,n)) - \alpha^{\downarrow}(L(m,n) + l(n)) \nonumber \\ 
&\le& \sup_{v \ge 0} \{ \beta^{\uparrow}(v) - \alpha^{\downarrow}(v + l(n))\} \label{eq-9}
\end{eqnarray} 
which completes the proof. 
\end{proof}

Since $\alpha^{\downarrow}$ is non-decreasing, i.e. (P1), Corollary \ref{th-1-c1} follows immediately from Theorem \ref{th-1}.

\begin{corollary}\label{th-1-c1}
The delay $D(n)$ of any packet $n (\ge 1)$ and the maximum packet delay $D^{max}$ are both upper-bounded by 
\begin{eqnarray}
 \sup_{v \ge 0}\{
    \beta^{\uparrow}(v) - \alpha^{\downarrow}(v+l^{min} ) 
    \}  
\end{eqnarray}
where $l^{min}$ denotes the minimum packet length. 
\end{corollary}

\subsection{Exemplification}\label{sec-4-2}

To demonstrate the potential improvement of the new bound over the classical bound, we consider in the following that the arrival curves are of the leaky-bucket type \cite{NetCal} and the service curves are of the rate-latency type \cite{NetCal}. 

\begin{proposition}\label{t1-p1}
    For a FIFO system, suppose the input has leaky-bucket arrival curve $\alpha(t) = \rho t + \sigma$, and the service has rate-latency service curve $\beta(t) = R(t-T)^+$, where $ (x)^{+}\equiv \max \{x, 0\}$. If $\rho \le R$, the delay $D(n)$ of any packet $n (\ge 1)$ satisfies:  
\begin{equation}\label{eq-new1}
D(n) \le \frac{\sigma}{R} + T - \frac{l(n)}{R} 
 \end{equation}
\end{proposition}

\begin{proof}
Consider any packet $n (\ge 1)$ and apply $\alpha(t) = \rho t + \sigma$ and $\beta(t) = R(t-T)^+$ to Theorem \ref{th-1}. The lower inverse function $\alpha^{\downarrow}(v+l(n))$ of  $\alpha(t)=\rho t + \sigma$ is: 
\begin{eqnarray}
        \alpha^{\downarrow} (v + l(n)) &=& \inf\{t \ge 0: \rho t + \sigma \ge v + l(n) \} \nonumber\\
        &=&  \left(\frac{v + l(n) -\sigma}{\rho} \right)^{+} \nonumber
\end{eqnarray}
and the upper inverse function $\beta^{\uparrow}$ of $\beta(t) = R(t-T)^+$ is $\beta^{\uparrow}(v) = \inf\{t (\ge 0):  R(t-T)^+ > v \}= \frac{v}{R}+T$. 
Hence
\[
   \beta^{\uparrow}(v) - \alpha^{\downarrow}(v+l(n))  = 
\begin{cases}
    \frac{v}{R} + T - \frac{v+l(n)-\sigma}{\rho}, & \text{if } v  + l(n)\geq \sigma\\
    \frac{v}{R}+T,              & \text{if }  v  + l(n) < \sigma
\end{cases}
\]
It can be shown that the supremum of $ \beta^{\uparrow}(v) - \alpha^{\downarrow}(v+l(n))$ is $\frac{\sigma - l(n)}{R}+T$, which completes the proof. 
\end{proof}

From Proposition \ref{t1-p1}, since  $\alpha^{\downarrow}$ is non-decreasing, Corollary \ref{t1-p1-c1} can be verified. 

\begin{corollary}\label{t1-p1-c1}
For the system considered in Proposition \ref{t1-p1}, the delay $D(n)$ of any packet $n (\ge 1)$ and the maximum packet delay $D^{max}$ are both upper-bounded by 
\begin{equation}\label{cmp-new1}
\frac{\sigma}{R} + T - \frac{l^{min}}{R} \equiv \Delta^{New}. 
 \end{equation}
\end{corollary}

As a comparison, for the system considered in Proposition \ref{t1-p1}, the classical delay bound from virtual delay analysis becomes 
\begin{equation}\label{cmp-nc}
    \frac{\sigma}{R}+T \equiv \Delta
\end{equation}
with which, we can rewrite $\Delta^{New}$ as 
$$\Delta^{New} = \Delta - \frac{l^{min}}{R}.$$
Clearly the new bound $\Delta^{New}$ is tighter, since in real networks, the length of a packet is always lower-bounded e.g. by the header size. This leads to the following remark. 

\vspace{6pt}
\noindent \textbf{Remark 1:} {\em The improvement of $\Delta^{New}$ over $\Delta$ is strictly positive in practical systems.}

\section{Case Study} \label{sec-5}
In this section, we conduct a case study to compare the delay bounds quantitatively. It is based on an example configuration for a time-aware Talker in time-sensitive networking (TSN), as described in Annex U: TSN Configuration Examples, \cite{8021Q}. 

\subsection{The Case}
The transmission rate of the Talker is $c=100$ Mb/s. The Talker has two queues, a TSN queue and a non-TSN queue, where the TSN queue has non-preemptive strict priority over the non-TSN queue. When the TSN queue is enabled, a non-TSN frame can interfere and this interference is up to $\tau_{interference} = 123.36$ $\mu s$, assuming maximum frame size 1522 bytes for the non-TSN queue. 

The Talker supports two TSN streams, Stream J and Stream K. Both generate frames periodically. For Stream J, the frame size is $l_J = 1500$ bytes and the interval is 500 $\mu s$. For Stream K, the frame size is $l_K = 1000$ bytes, and the interval is also 500 $\mu s$. The two streams form the input to the TSN queue, and hence for the TSN queue, $l^{max} = 1500$ bytes and $l^{min} = 1000$ bytes. 
In addition, $l_{J}^{max} = l_{J}^{min}= 1500$ bytes and $l_{K}^{max} = l_{K}^{min}= 1000$ bytes.

For this case, it can be verified that Stream J has leaky-bucket arrival curve $\alpha_J(t) = \rho_J t + \sigma_J$ with $\rho_J= 24$ Mb/s and $\sigma_J=l_J=1500$ bytes, and Stream K has arrival curve $\alpha_K(t) = \rho_K t + \sigma_K$ with $\rho_K= 16$ Mb/s and $\sigma_K=l_K = 1000$ bytes. The aggregate input to the TSN queue has arrival curve  $\alpha(t) = \rho t + \sigma$ with $\rho = \rho_J + \rho_K = 40$ Mb/s and $\sigma=l_J+l_K=2500$ bytes. In addition, the system provides to the TSN-queue latency-rate service curve $\beta(t) = R(t-T)^{+}$ with $R=c= 100$ Mb/s and $T=\frac{l^{max}}{c} + \tau_{interference} = 243.36$ $\mu s$ (e.g. see Sec. 2.4.2, \cite{Chang00}). With these, the classical bound $\Delta$ and the new bound $\Delta^{New}$ are readily calculated. 

\subsection{Related Delay Bounds}

Note that the case description contains more specific  information than $\alpha$ and $\beta$. This information can be exploited in delay bound analysis. 
For instance, Proposition \ref{t1-p1} from Theorem \ref{th-1} can be generalized to Corollary \ref{t1-p2} by exploiting the information of each individual flow, whose proof follows from focusing only on packets of the considered flow in (\ref{eq-new1}). 

\begin{corollary}\label{t1-p2}
    If the input consists of multiple flows, each flow $i$ has arrival curve $\alpha_i=\rho_i t + \sigma_i$ and $\sum_i \rho_i \le R$, the delay $D(n_i)$ of any packet $n_i$ in flow $i$ satisfies: 
\begin{eqnarray}
  D(n_i) &\le& \frac{\sum_i \sigma_i}{R} + T - \frac{l(n_i)}{R} \nonumber   
\end{eqnarray}
    and the maximum packet delay of flow $i$ is upper-bounded by 
\begin{equation}\label{cmp-new2}
\Delta - \frac{l_i^{min}}{R} \equiv \Delta_i^{NEW}. 
 \end{equation}
 where $l_i^{min}$ denotes the minimum packet length of flow $i$. 
\end{corollary}

In addition, the description contains information that enables the application of delay bounds from \cite{Ehsan19} \cite{Ehsan23}. Specifically it is proved in \cite{Ehsan19} that, for a FIFO element with rate-latency service curve $R(t-T)^+$, if the transmission rate $c (\ge R)$ is known, the maximum packet delay of any flow $i$ is upper-bounded by: 
\begin{equation} \label{cmp-ehsan19}
    \Delta- l_i^{min}(\frac{1}{R}-\frac{1}{c}) \equiv \Delta^{\text{\cite{Ehsan19}}}_i. 
\end{equation}
If the TSN traffic specification of each flow is also known, it is further proved in \cite{Ehsan23} that the maximum packet delay of any flow $i$ is upper-bounded by:  
\begin{equation}\label{cmp-ehsan23}
    \Delta- l_i^{max}(\frac{1}{R}-\frac{1}{c}) \equiv \Delta^{\text{\cite{Ehsan23}}}_i
\end{equation}
where $l_i^{max}$ denotes the maximum packet length of flow $i$. 

Comparing (\ref{cmp-new2}) with (\ref{cmp-ehsan19}) and (\ref{cmp-ehsan23}), it is clear that $\Delta_i^{NEW} = \Delta_i^{\text{\cite{Ehsan19}}} - \frac{l_i^{min}}{c} $ and $\Delta_i^{NEW} = \Delta_i^{\text{\cite{Ehsan23}}} - \frac{l_i^{min}}{R} + l_i^{max}(\frac{1}{R}-\frac{1}{c})$. They imply that $\Delta_i^{NEW}$ is always tighter than $\Delta_i^{\text{\cite{Ehsan19}}}$ but its tightness in comparison with $\Delta_i^{\text{\cite{Ehsan23}}}$ is influenced by $l_i^{min}$, $l_i^{max}$, $R$ and $c$. In addition, comparing (\ref{cmp-new2}) with (\ref{cmp-new1}), it is clear $\Delta_i^{NEW} = \Delta^{New} - \frac{l_i^{min}-l^{min}}{R}$. Since $l^{min}=\min_i \{ l_i^{min}\} \le l_i^{min}, \forall i$, we always have $\Delta_i^{NEW} \le \Delta^{New}$. It is worth highlighting that $\Delta_i^{NEW}$ does not rely on the additional information, e.g. $c$, required for  $\Delta_i^{\text{\cite{Ehsan19}}}$ and $\Delta_i^{\text{\cite{Ehsan23}}}$.  
In addition, $\Delta^{New}$ does not rely on the  information of individual flows required for $\Delta_i^{NEW}$. To give a more direct impression, the bounds are compared quantitatively for the case considered. 

\subsection{Quantitative Comparison} 

Table \ref{tb-cmp} summarizes and compares the delay bounds for the considered case, which are arranged in two categories. One contains the new bound $\Delta^{New}$ and the classical bound $\Delta$ that do not require or make use of traffic information of the constituent flows. The other category contains the new bound $\Delta_i^{NEW}$ and the bound $\Delta_i^{\text{\cite{Ehsan19}}}$ from \cite{Ehsan19} and the bound $\Delta_i^{\text{\cite{Ehsan23}}}$ from \cite{Ehsan23}, which make use of information of the constituent flows. 

\begin{table}[th!]
\caption{Comparison of delay bounds (in $\mu s$) for Streams J and K} \label{tb-cmp}
\centering
%\begin{center}
\begin{tabular}{|l|c|c||c|c|c|}
\hline
Stream & $\Delta^{New}$ &  $\Delta$ &  $\Delta_i^{NEW}$ & $\Delta_i^{\text{\cite{Ehsan19}}}$ & $\Delta_i^{\text{\cite{Ehsan23}}}$\\ \hline\hline  
 J & 363.36 & 443.36   & 323.36 & 443.36  & 443.36  \\ \hline 
 K & 363.36 & 443.36   & 363.36 & 443.36  & 443.36 \\ \hline 
\end{tabular}%
%\end{center}
\end{table}

As shown in Table \ref{tb-cmp}, the two new delay bounds $\Delta^{New}$ and $\Delta_i^{NEW}$ are more than $25\%$ tighter in their categories. 
In particular, Table \ref{tb-cmp} demonstrates that the classical bound $\Delta$ can be significantly conservative, which confirms the finding in Theorem \ref{th-0}.
In addition, it shows that with information about individual flows, $\Delta_i^{NEW}$ can improve $\Delta^{New}$ particularly for Stream J. Furthermore, $\Delta_i^{NEW}$ also exhibits significant improvement over $\Delta_i^{\text{\cite{Ehsan19}}}$ and $\Delta_i^{\text{\cite{Ehsan23}}}$ for the considered case. 

Note that both streams have the frame generation interval 500 $\mu s$. This implies that within each such interval, only two frames are generated, one from each stream. In addition, all bounds less than 500 $\mu s$ imply that the two frames finish before the next such frames are generated. With these, the worst case is that both frames arrive at the same time just before the transmission of a non-TSN frame, and the worst-case delay becomes  
$\frac{l_J}{c}+\frac{l_K}{c}+T_{interference} = 323.36$ $\mu s$. 
Comparing with Table \ref{tb-cmp}, while all bounds are valid, most are not tight. Recall that $\Delta$, $\Delta^{\text{\cite{Ehsan19}}}$ and $\Delta^{\text{\cite{Ehsan23}}}$ have been shown to be tight under their conditions. We hence remark the following. 

\vspace{6pt}
\noindent \textbf{Remark 2:} {\em A delay bound that is tight under certain conditions may not be tight if additional information is available.}

\section{Conclusion}\label{sec-6}
We revisited delay bound analysis in network calculus and showed that for a FIFO system, the maximum packet delay is upper-bounded by the maximum virtual delay.
Building on this insight, we developed an analysis that directly targets packet delay without relying on virtual delay and derived a new packet delay bound that requires no additional assumption beyond arrival and service curves. 
We also specialized the new bound for a system with leaky-bucket arrival curve and rate-latency service curve, showing that the new bound is strictly tighter. 
A TSN case study demonstrated that the new bound can produce a reduction of more than 25\% compared to the classical virtual-delay-based bound. In addition, the case study indicates that the tightness of a delay bound depends on the available information and cannot be universally generalized.

\bibliographystyle{IEEEtran}
%\bibliography{nc-qt}

\begin{thebibliography}{1}

\bibitem{Chang00}
C.-S. Chang, \emph{Performance Guarantees in Communication Networks}. 
  Springer-Verlag, 2000.

\bibitem{NetCal}
J.-Y. {Le Boudec} and P.~Thiran, \emph{Network Calculus: A Theory of Deterministic Queueing Systems for the Internet}.  
  Springer-Verlag, 2001.

\bibitem{DNC}
A.~Bouillard, M.~Boyer, and E.~{Le Corronc}, \emph{Deterministic Network
  Calculus: From Theory to Practical Implementation}.\hskip 1em plus 0.5em
  minus 0.4em\relax Wiley-ISTE, 2018.

\bibitem{Zhao22}
L.~Zhao, P.~Pop, and S.~Steinhorst, ``Quantitative performance comparison of
  various traffic shapers in time-sensitive networking,'' \emph{IEEE
  Trans. Network and Service Management}, vol.~19, no.~3, pp.
  2899--2928, 2022.

\bibitem{Ehsan19}
E.~Mohammadpour, E.~Stai, and J.-Y. Le~Boudec, ``Improved delay bound for a
  service curve element with known transmission rate,'' \emph{IEEE Networking
  Letters}, vol.~1, no.~4, pp. 156--159, 2019.

\bibitem{Ehsan23}
------, ``Improved network-calculus nodal delay-bounds in time-sensitive
  networks,'' \emph{IEEE/ACM Trans. Networking}, vol.~31, %no.~6, 
  pp. 2902--2917, 2023.

\bibitem{JL17}
J.~{Liebeherr}, ``Duality of the max-plus and min-plus network calculus,''
  \emph{Foundations and Trends in Networking}, vol.~11, pp. 139--282,
  2017. %no. 3-4, 

\bibitem{8021Q}
``{IEEE} standard for local and metropolitan area networks--bridges and bridged
  networks,'' \emph{IEEE Std 802.1Q-2022 (Revision of IEEE Std 802.1Q-2018)},
  pp. 1--2163, 2022.

\end{thebibliography}

\vfill

\end{document}